\documentclass[letterpaper,12pt]{amsart}
\usepackage{times}
\usepackage[margin=1.0in]{geometry}
\usepackage{graphicx}
\usepackage{amsmath,amsthm,amsbsy,amsfonts,amssymb,dsfont}

\usepackage{amssymb}
\usepackage{mathrsfs}

\newcommand{\suppress}[1]{}

\theoremstyle{plain}
\newtheorem{theorem}{Theorem}
\newtheorem{lemma}[theorem]{Lemma}

\newtheorem{prop}[theorem]{Proposition}
\theoremstyle{definition}

\newtheorem{example}[theorem]{Example}
\theoremstyle{remark}
\newtheorem{remark}[theorem]{Remark}

\newcommand{\rset}{\mathbb R}

\numberwithin{equation}{section}

\newcommand{\be}{\begin{equation}}
\newcommand{\ee}{\end{equation}}
\newcommand{\bea}{\begin{eqnarray}}
\newcommand{\eea}{\end{eqnarray}}
\newcommand{\bean}{\begin{eqnarray*}}
\newcommand{\eean}{\end{eqnarray*}}
\newcommand{\da}{\dagger}

\DeclareMathOperator{\Span}{span}
\DeclareMathOperator{\Tr}{Tr}

\DeclareMathOperator{\LAP}{\mathscr{L}}

\newcommand{\q}{{q}}

\usepackage{color}

\newcommand{\ket}[1]{|#1\rangle}
\newcommand{\bra}[1]{\langle#1|}
\newcommand{\Lap}{{L_C}}
\newcommand{\RG}{E}
\newcommand{\tpsi}{\tilde{\psi}}
\newcommand{\tW}{\tilde{W}}
\newcommand{\tw}{\tilde{w}}
\newcommand{\tL}{\tilde{L}}
\newcommand{\hb}{\hat{b}}

\newcommand{\Pp}{P}

\begin{document}
\title[The Invisible Hand of Laplace]{The Invisible Hand of Laplace: the Role of Market Structure in Price Convergence and Oscillation}

\author{Yuval Rabani}\thanks{YR was supported in part by ISF grant 856-11, by BSF grant 2012333, and by I-CORE Algo.
LJS was supported in part by NSF grant 1319745 and by BSF grant 2012333.
Addresses: The Hebrew University of Jerusalem, The Rachel and Selim Benin School of
Computer Science and Engineering, 9190401 Jerusalem, Israel, {\tt yrabani@cs.huji.ac.il};
Caltech, Engineering and Applied Science MC305-16, Pasadena CA 91125, USA, {\tt schulman@caltech.edu}. }
\author{Leonard J. Schulman}

\begin{abstract} 
A fundamental question about a market is under what conditions, and then how
rapidly, does price signaling cause price equilibration. Qualitatively, this ought to depend on how well-connected the market is. We address this question quantitatively for a certain class of Arrow-Debreu markets with continuous-time proportional t\^{a}tonnement dynamics. We show that the algebraic connectivity of the market determines the effectiveness of price signaling equilibration. This also lets us study the rate of external noise that a market can tolerate and still maintain near-equilibrium prices.
 \end{abstract}
\maketitle

\section{Introduction}
In a free market, the rise or fall of a price signals excess demand or supply; if the dynamics of price adjustment work well, signaling can enable goods to clear and prices to equilibrate.
A fundamental question about a market is therefore
under what conditions, and
then how rapidly, does price signaling cause price equilibration. It is a stubborn question however how best to model price dynamics. A
variety of both t\^{a}tonnement (non-trading)~\cite{Walras74}
and trading processes have
been offered as models. The question has been particularly studied in the context of Arrow-Debreu markets~\cite{McKenzie54,ArrowD54}: it was shown early on that
continuous-time
t\^{a}tonnement
converges to an equilibrium if goods are gross substitutes~\cite{ABH59} but otherwise, as Scarf~\cite{Scarf60} exhibited, it may not.
This spurred the study of alternative forms of t\^{a}tonnement~\cite{Scarf67,Kuhn68,Smale76,vanderLaanT87,Kamiya90,Herings97,JoostenT98,beiGH15}, which were shown to have stronger convergence properties.
Neither t\^{a}tonnement
 nor existing trading models (see below), however, are fully reductionist theories; that is, we do not have a
model of individual strategic
transactions from which emerges at the market level an Arrow-Debreu market with the stated dynamics and which equilibrates to the given endowments.

Despite this difficulty in the theory, price signaling in practice certainly works to stabilize prices.
In a typical market, over long periods,
prices oscillate with mild amplitude within a near-equilibrium zone and goods roughly clear.
After any shock, the market restores itself to a possibly new near-equilibrium zone. This is after all the rationale for studying equilibrium theories. Moreover, in the last 15 years laboratory evidence has accumulated in support of t\^{a}tonnement dynamics~\cite{Plott00,AndersonPSG04,HirotaHPR05,Gjerstad07,CrockettOP11} even in markets such as Scarf's or Gale's~\cite{Gale63} where it makes surprising predictions. In short,
there is no question but that the ``invisible hand'' of price signaling does operate within free markets. A good economic theory ought to capture quantitative aspects of its operation. Such quantification should include information about rates of convergence, at least in the vicinity of a stable equilibrium; as well as, ideally, also information about the typical
out-of-equilibrium modes that may empirically be observed, since modes which do not damp rapidly can be expected to be continually present due to external stimulation.
We emphasize that a true quantitative theory cannot be one-sided (bounding the rate of convergence only from above or only from below)---if we wish to compare two markets, or understand whether oscillations are due to stimulation or are long-lived modes of the system, we need both \textit{upper} and \textit{lower} bounds.

In this paper we take a step toward such quantification. Given the already-discussed challenge of modeling markets out of equilibrium,
our approach is very conservative: we work entirely within the Arrow-Debreu market framework, and restrict ourselves to the arguably simplest price dynamic, continuous-time proportional-t\^{a}tonnement (CTPT)
as in Samuelson~\cite{Samuelson41}. For an overview of this topic see~\cite{McKenzie02,Mukherji02}, and for experimental evidence for this particular class of dynamics in trading markets, see~\cite{HirotaHPR05,CrockettOP11}. (It is notable that the t\^{a}tonnement process is predictive of trading dynamics despite formally involving no trade.)
We restrict ourselves to a setting where these dynamics converge, namely, the gross substitutes regime~\cite{ABH59} (see also~\cite{AH58,Uzawa61}) and, since we are asking for rather precise results, we assume that utilities take on a predictable form; for this purpose we again make the most standard assumption possible, namely that agents have CES utilities.
 Furthermore, since we wish to establish both
upper and lower
bounds on rate of convergence,
 we restrict ourselves entirely to system behaviour in a neighborhood of equilibrium. It may be possible to extend our results to the entire space of prices, but one would have to argue that this is economically well-motivated, since the simple CTPT rule has many consistent extensions away from equilibrium. We follow Tolstoy's famous dictum about families but, not being novelists, focus our attention on \textit{happy} markets, which for us means, those near equilibrium.

We have, therefore, eliminated almost all parameters available to tweak the market model, except the parameter we wish to study: the \textit{connectivity structure} of the market. One expects that a market in which all participants interact pairwise will have relatively sharp and rapid reactions to price imbalances; whereas one in which participants have only indirect effects on most others,
will adjust more moderately and slowly. We shall quantify this phenomenon fairly precisely through
the spectrum of a Laplacian matrix that is derived from the connectivity structure of the market. Specifically, the second eigenvalue of this matrix, also known as the \textit{algebraic connectivity} of the network, will be shown to determine, in most cases, the convergence rate of market dynamics.
If the market is continually buffeted by external noise at a fixed rate, then the convergence rate determines in turn the typical distance of prices from equilbrium, and this too will be quantified.
 More precise statement awaits the definitions in the next section.

The Laplacian, also called the heat kernel, originates in Physics (in the study of heat diffusion), in Probability (in the study of random walk), and in Electrical Engineering (in the study of resistive networks); in the last decades it has found wide applicability in Finance, Combinatorics and Theoretical Computer Science. We are not, however, aware of other applications of it to the study of  structure of economic markets.

\subsection{Other related work}
Before moving on to our development, we mention that not all study of market dynamics has been on non-trading processes; out-of-equilibrium trade models have also been developed and shown to converge, such as the Edgeworth process~\cite{Uzawa62} or the Hahn process~\cite{Hahn62,HahnN62}; however, these are less relevant to our study for several reasons, including that the former requires coordination of large coalitions; and the latter converges to an equilibrium that generally does not agree with the specified endowments. For more see~\cite{Fisher83}. Simply put, although it would be desirable to pursue our topic in a trading process, no model with all the needed properties has yet been found. Perhaps, as we focus our study in a neighborhood of equilibrium, the distinctions between models are however not so significant.

The history of work on the stability of t\^atonnement begins with Hicks~\cite{Hicks39}, who discussed the local stability of t\^atonnement under some conditions on the market. Samuelson~\cite{Samuelson41,Samuelson44} showed that the Hicksian conditions are neither sufficient nor necessary for stability. Metzler~\cite{Metzler45} showed that the Hicksian conditions are sufficient for stability in the case of gross substitutes utilities. 

We also mention a recent line of work  in theoretical computer science on t\^atonnement~\cite{CMV05,FGKKS08,CF08,BDX11,CCR12,CCD13,ARY15} and related processes~\cite{zhang09}. These papers consider Fisher markets with CES utilities (not necessarily gross substitutes) and also some cases beyond this class of utility functions. Fisher markets are special-case Arrow-Debreu markets, in which the total initial endowments of goods is split among all the participants in the same proportions for all goods. Therefore, changing the relative prices of the goods does not change the relative purchasing power of the participants, only their desired consumption basket. Such markets are attractive from a computational perspective, as there are efficient algorithms to compute equilibrium prices and allocations. The above-mentioned papers propose several
discrete time interpretations of CTPT, and establish global
 {\em upper}
bounds on the
convergence time of discrete-time t\^atonnement. The highlight of this line of work from our perspective is the paper~\cite{CCD13} that relates discrete-time t\^atonnement to the convex optimization method of gradient descent, and upper bounds the  convergence time across the CES spectrum (including the complementary case of $\rho < 0$).
The results in these papers are quite incomparable with ours. On the positive side, the results
hold for a wider range of utility functions and apply to the entire space of prices, not only near equilibrium.
On the negative side, the results apply only to Fisher markets; and the most essential distinction is that these works do not
address our main question of quantifying the rate of convergence in terms of the market structure---in fact their framework does not
consider the market structure at all, and they only provide worst-case one-sided (upper) bounds on the convergence time.

In a different vein, researchers have been interested in other market structure effects: as these works do not directly impact ours, we do not attempt a survey, but only provide a few pointers:~\cite{krantonM01} looks at markets in which buyer-seller pairs can trade only along established links, at the incentives to form such links, and at the efficiency of trade in such networks;~\cite{kakadeKO04} considers Arrow-Debreu markets in which, again, direct trade can occur only along established links, thus enabling the same commodity to have different prices in different places; and~\cite{blumeEKT09} looks at markets in which buyer-seller pairs can only interact through intermediary traders, and studies the power of these traders and how equilibrium prices are affected by the connectivity structure.

\section{The model} \label{model}

We consider an Arrow-Debreu market in which each participant $i$ is endowed with a quantity $s_i>0$ of a unique perfectly divisible good, also denoted by $i$. There are $n \geq 2$ participants, and participant $i$ values an allocation $x_{ij}$ of the goods $j$ according to a CES utility function:
\[ u_i(x) = \left(\sum_j (c_{ij}x_{ij})^\rho\right)^{1/\rho} \]
The parameter $\rho$ is shared by all and is in the gross substitutes regime, $\rho \in (0,1)$.
Simply by rescaling units, we may suppose that all supplies $s_i=1$; this entails replacing any $c_{ij}$ in the original market by $c_{ij}s_j$ in the rescaled market.
The coefficients $c_{ij}$ can vary widely but cannot be entirely general; they must satisfy the following three conditions. (a) Nonnegativity: $c_{ij} \geq 0$.
(b) Connectedness of the market: for every $i,j$ there are $i=i_0,i_1,\ldots,i_k=j$ s.t.\ $\prod_{\ell=1}^k c_{i_{\ell-1}i_\ell}>0$.
 (c) Circulation-free: let $i_0,i_1,\ldots,i_k=i_0$ be any cycle through the vertices. Then $\prod_{\ell=1}^k c_{i_{\ell-1}i_\ell}=\prod_{\ell=1}^k c_{i_\ell i_{\ell-1}}$. (Observe that these conditions are not affected by the rescaling of the supplies.)

It is a consequence of (b) and (c) that the coefficients are weakly undirected in the sense that $c_{ij}=0$ if and only if $c_{ji}=0$. Consequently,
it will be useful to conceive of the participants as vertices of a graph, with $i,j$ connected by an undirected edge if $c_{ij}>0$. There may be self-loops in this graph.
Let $C_{ij}=c_{ij}^\delta$. Let $U$ be the adjacency matrix of this graph (the ``unweighted'' market graph), defined by $U_{ij}=1$ if $C_{ij}>0$, $U_{ij}=0$ if $C_{ij}=0$.

It will simplify expressions to replace the customary parameter $\rho$ by $\delta=\rho/(1-\rho)$ (note that $1/(1-\rho)$ is the {\em elasticity of substitution}, indicating the extent that desired goods serve as substitutes for each other). In all theorems in this paper, $\delta$ is an arbitrary value in $ (0,\infty)$; this corresponds to $\rho$ being an arbitrary value in $ (0,1)$.

In the Arrow-Debreu model, at prices $p_i$, not all $0$, participant $i$ has budget $b_i=p_i$,
 which is then allocated to goods $j$ so as to optimize basket utility; this results in the following demand by $i$ for $j$:
\be d_{ij}(p)=
\frac{p_i C_{ij}}{p_j^{1+\delta} \sum_k C_{ik}/p_k^{\delta}}
\label{demandG} \ee
(the argument $p$ will generally be understood and we will abbreviate to $d_{ij}$). 
Observe that these demands are invariant under rescaling the coefficients $C_{ij}$ by any positive multipliers $\alpha_i$. The dynamics (to be described in Section~\ref{sec:dr}) depend only on the demands and supplies, so we from now on rescale the $C_{ij}$ such that for every $i$, $\min_{j: C_{ij}>0} C_{ij}=1$ (and of course some $C_{ij}>0$). This rescaling, too, preserves the conditions (a,b,c).

In summary, for any given $\delta$ the market is fully specified by the data $C$.

We let $\Pp$ be the following function of prices $p$: $\Pp_i(p)=\sum_k \frac{C_{ik}}{p_k^\delta}$. So
$  d_{ij}=
\frac{p_i C_{ij}}{p_j^{1+\delta} \Pp_i(p)}$.

Let $d_j=\sum_i d_{ij}$ denote the total demand for good $j$. Prices are in equilibrium if $d_j=1$ for all $j$. Throughout the paper we use $r$ to denote a vector of equilibrium prices, and $R_i=\Pp_i(r)=\sum_k \frac{C_{ik}}{r_k^\delta}$; then the equilibrium condition can be rewritten as the following system of equations:
\be 1=\frac{1}{r_j^{1+\delta}} \sum_{i } \frac{r_i C_{ij}}{R_i} \label{demand-ruleG2} \ee
with
\[ d_{ij}(r) =\frac{1}{r_j^{1+\delta}} \frac{r_i C_{ij}}{R_i}
\]
being the equilibrium demand by $i$ for $j$.

\section{The dynamics and our results} \label{sec:dr}
In the subsequent Section~\ref{seq:eq} we justify
existence, uniqueness, and certain properties of the equilibrium vector of prices, denoted $r$. The focus of the paper is dynamics in a neighborhood of $r$; we now describe those dynamics and our results.

\subsection{Dynamics}
As indicated earlier, the dynamics are
proportional t\^{a}tonnement in continuous time,
specifically, for a price vector $p$,
\be \dot{p}_j = (d_j-1) p_j \label{dyn:p} \ee
and we study these dynamics in a perturbation of equilibrium, $p_j=r_j e^{\alpha_j}$ for small $\alpha$. Then
\be \dot{\alpha}_j = \frac{\dot{p}_j}{p_j} \label{dyn:alpha} \ee
If $c_j$ is a function of $1 \leq j \leq n$, $\ket{c}$ or $\ket{c_j}\in \rset^n$ denotes a column vector with $j$'th entry $c_j$; if $c$ is a scalar, $\ket{c}\in \rset^n$ denotes a column vector with all entries $c$. At equilibrium, by definition, $\dot{p}=\ket{0}$, so, referring to the dynamics of Eq.~\eqref{dyn:alpha}, we have
 $\dot{\alpha}=\ket{0}$ at $\alpha=\ket{0}$.
Consequently
\[ \dot{\alpha}_j = \sum_i
 \left.\frac{\partial d_j}{\partial \alpha_i} \right|_{\ket{0}} \cdot \alpha_i 
\]
From now on we abbreviate $D_{ji}=\left. \frac{\partial d_j}{\partial \alpha_i} \right|_{\ket{0}} $,
so the above becomes 
\be \ket{\dot{\alpha}} = D \ket{\alpha} .\label{alphaD} 
\ee
Any scaling of $r$ is an equilibrium, so
\be \ket{0} = D \ket{1}. \label{alpha-ker} \ee

\subsection{Results} \label{subsec:results} It is very familiar in dynamical systems that the controlling parameter in dynamics of the form Eq.~\eqref{alphaD} is the maximum real part of any eigenvalue of $D$; the system is unstable if the maximum is positive. As we will see below, the eigenvalues of our operator $D$ are real, and so may be indexed
$\lambda_{\downarrow 1}(D) \geq \ldots \geq \lambda_{\downarrow n}(D)$.
As we see from Eq.~\eqref{alpha-ker}, however, one of these eigenvalues is $0$. This particular eigenvalue is irrelevant to our considerations: it merely expresses that the problem is scale-invariant in the prices, and so equilibrium prices form a ray. We will express this in the dynamics by writing $\alpha(t)$ for the price vector perturbation at time $t$, and decomposing this into a part proportional to $\ket{1}$ which is unchanging, and a remainder, as follows: $\alpha(t)=\ket{c} + \bar{\alpha}(t)$ for some real $c$. The key will be to understand the dynamics of $\bar{\alpha}(t)$.

We will also show below that all eigenvalues are nonpositive. Consequently, the quantity of true interest to us is $\lambda_{\downarrow 2}(D)$. This quantity has the following meaning: there is an invertible matrix $B$ (to be introduced below) such that, with $\| x\|_B := \sqrt{\bra{x} B^\da B \ket{x}}$,

(i) For all $\bar{\alpha}$,
$\|\bar{\alpha}(t) \|_B \leq \|\bar{\alpha}(0) \|_B \cdot e^{\lambda_{\downarrow 2}(D)t}$.

(ii) For some $\bar{\alpha} \neq 0$,
$\|\bar{\alpha}(t) \|_B = \|\bar{\alpha}(0) \|_B \cdot e^{\lambda_{\downarrow 2}(D)t}$.

 Thus $\lambda_{\downarrow 2}(D)$ can be thought of as the \textit{damping rate} for perturbations from equilibrium; alternatively,
 $\frac{-\log 2}{\lambda_{\downarrow 2}(D)}$ is the \textit{convergence time} of the market, the time in which any perturbation from equilibrium will halve in norm (in a certain preferred basis).

Before giving the results it is necessary to introduce formally the notion of the Laplacian (or sometimes called the normalized Laplacian) of a weighted graph. A weighted graph is, for this purpose, a real symmetric matrix $A$ with nonnegative entries, in which no row is $0$. Rows and columns are indexed by vertices of the graph and we say that $i \sim j$, $i$ and $j$ are connected by an edge, if $A_{ij}>0$. The Laplacian corresponding to $A$ is the matrix $\LAP(A)$ defined as follows. Define $a=\sigma(A)$ to be the diagonal matrix with entries $a_{ii}> 0$ for $a_{ii}^2=\sum_j A_{ij}$. Then $\LAP(A) = I - a^{-1} A a^{-1}$. It is well known that for any $A$,
$0=\lambda_{\uparrow 1}(\LAP(A)) \leq \ldots \leq \lambda_{\uparrow n}(\LAP(A)) \leq 2$, and that
 the rank of the kernel of $\LAP(A)$ is the number of connected components of $A$. Note that $\LAP$ is homogeneous of degree $0$, i.e., invariant to scaling of its argument.

We will be discussing the Laplacians of several different weighted graphs. The most important, to be called $\Lap$, quantifies exactly the damping rate of the market. The following proposition is therefore central to the paper.

\begin{prop}[Damping rate characterization] \label{L-expansion}
 Let
\be \q(\delta,\lambda)=-(1+2\delta)\lambda+\delta \lambda^2. \label{q-def} \ee
There is a diagonal matrix $B$ such that
\[ D=B^{-1} \cdot q(\delta,\Lap) \cdot B \]
and therefore the damping rate is 
$\lambda_{\downarrow 2}(D) = \lambda_{\downarrow 2}(q(\delta,\Lap))$.
\end{prop}

It should be said however that $\Lap$
has a rather complicated dependence on $C$; and that $C$ itself may be hard to know precisely. We will therefore devote attention to obtaining plainer bounds. 

There is one very special case however in which $\Lap$ is easy to write down: the ``uniform" case,
 in which all coefficients $C_{ij}$ are either $0$ or $1$ and every participant has the same ``degree'', that is, there is an integer $\Delta>1$ s.t.\ for all $i$, $|\{j: C_{ij}=1\}|=|\{j: C_{ji}=1\}|=\Delta$. In this case (a) equilibrium prices are uniform, (b) $\Lap=\LAP(U)$. The calculations in this case are simpler than in general and allow for the following statement.
\begin{theorem}[Bounds for uniform markets] In the special case just defined,
\be
\q(\delta,\lambda_{\uparrow 2}(\LAP(U)))
\leq
 \lambda_{\downarrow 2}(D)
\leq
 \max\{q(\delta,\lambda_{\uparrow 2}(\LAP(U))),-2\}
 \ee
with the lower bound being tight if $\delta \leq 1/2$ or in other cases discussed in Section~\ref{sec:genbds}.
\label{thm:unif} \end{theorem}

The general case is given in Theorem~\ref{precise}: it is a strict generalization of Theorem~\ref{thm:unif}, given in terms of the aforementioned Laplacian $\Lap$, whose definition will be given in Section~\ref{sec:dyn}.

\begin{theorem}[Bounds for more general markets]
$q(\delta,\lambda_{\uparrow 2}(\Lap ))
\leq
 \lambda_{\downarrow 2}(D)
\leq \max\{
q(\delta,\lambda_{\uparrow 2}(\Lap )), -2\}$.
\label{precise} \end{theorem}
As noted above, $\Lap$ is not easy to read off from the market data $C$.
 For this reason, we follow
Theorem~\ref{precise} with
 ``comparison theorems'' in which we provide weaker bounds in terms of more easily-obtained quantities. The main tool here is a lemma in spectral graph theory which ought to be known, but which we have not been able to find previously proven.
 First, a definition:

For two $n \times n$ weighted adjacency matrices $W,\tW$, let $\nu=\nu(W,\tW)=\left(\max_{i,j} \frac{W_{ij}}{\tW_{ij}} \right) \cdot \left( \max_{i,j} \frac{\tW_{ij}}{W_{ij}} \right)$, with the ratios taken as $1$ when numerator and denominator are $0$. (Thus $\nu\geq 1$, with $\nu=1$ only if $W,\tW$ are scale copies of the same matrix.)

\begin{lemma}[Laplacian stability] \label{lap:stab}
$\lambda_{\uparrow 2} (\LAP(\tW)) \leq \nu(W,\tW) \cdot \lambda_{\uparrow 2} (\LAP(W))$. This bound is best possible for all $\nu$. 
\end{lemma}
(In this lemma $W$ may be any weighted adjacency matrix, not necessarily our market graph, although that is how we apply the lemma.)
\begin{proof}
We may assume that $W$ is connected, otherwise apply the lemma separately to each connected component.

Note that there is always a $c>0$ s.t.\ $W_{ij}\leq \nu^{1/2} c\tW_{ij} \leq \nu W_{ij}$ for all $i,j$ (and this serves as an alternative definition of $\nu$). Recalling that $\LAP$ is invariant under rescaling of its argument, we may assume that $\tW$ has been scaled so that $W_{ij}\leq \nu^{1/2} \tW_{ij} \leq \nu W_{ij}$ for all $i,j$.

Let $w=\sigma(W)$ and $\tw=\sigma(\tW)$. Let $L=\LAP(W)$ and $\tL=\LAP(\tW)$.

It is well known that $\ker{L}=\Span \ket{w}$. Likewise $\ker{\tL}=\Span \ket{\tw}$. By the spectral theorem,
\[ \lambda_{\uparrow 2}(L)=\inf_{x: \bra{x}\ket{w}=0} \frac{\bra{x}L\ket{x}}{\bra{x}\ket{x}} \]
and applying the transformation $\ket{b}=w^{-1}\ket{x}$ we have
\be \lambda_{\uparrow 2}(L)=\inf_{b: \bra{b}\ket{w^2}=0} \frac{\bra{b}wLw\ket{b}}{\bra{b}w^2\ket{b}} \label{lamL} \ee
Note that
\[ \frac{\bra{b}wLw\ket{b}}{\bra{b}w^2\ket{b}}=
\frac{\sum_{i<j}W_{ij} (b_i-b_j)^2}{\sum_i w_{i}^2 b_i^2} =: R_W(b) \]
$R_W(b)$ is known as the Raleigh quotient of $b$ in $W$.

Let $b$ be a vector achieving Eq.~\eqref{lamL}, that is to say, a second eigenvector of $L$. So $\bra{b}\ket{w^2}=0$ and
 $\lambda_{\uparrow 2}(L)=\frac{\bra{b}wLw\ket{b}}{\bra{b}w^2\ket{b}}$.

  We use $b$ to produce a proxy $\hb$ for a second eigenvector of $\tL$:
\[ \ket{\hb}=\ket{b}-\ket{1} \frac{\bra{\tw^2}\ket{b}}{\bra{\tw}\ket{1}} \]
This satisfies the required
\[ \bra{\tw^2}\ket{\hb}=0. \]
So
\[ \lambda_{\uparrow 2}(\tL) \leq  R_{\tW}(\hb) =
\frac{\sum_{i<j}\tW_{ij} (\hb_i-\hb_j)^2}{\sum_i \tw_{i}^2 \hb_i^2}
=
\frac{\sum_{i<j}\tW_{ij} (b_i-b_j)^2}{\sum_i \tw_{i}^2 \hb_i^2} \]
Upper bounding the entries of $\tW$, we have
\[ \ldots
\leq
\nu^{1/2}
\frac{\sum_{i<j}W_{ij} (b_i-b_j)^2}{\sum_i \tw_{i}^2 \hb_i^2} \]
and lower bounding the entries of $\tw$, we have
\bean
\ldots &\leq &
\nu
\frac{\sum_{i<j}W_{ij} (b_i-b_j)^2}{\sum_i w_{i}^2 \hb_i^2}
\\ &=&
\nu
\frac{\sum_{i<j}W_{ij} (b_i-b_j)^2}{\sum_i w_{i}^2 b_i^2}
\frac{\sum_i w_{i}^2 b_i^2}{\sum_i w_{i}^2 \hb_i^2}
\\ &=& \nu R_W(b)
\frac{\sum_i w_{i}^2 b_i^2}{\sum_i w_{i}^2 \hb_i^2}
\\ &=& \nu \lambda_{\uparrow 2}(L)
\frac{\sum_i w_{i}^2 b_i^2}{\sum_i w_{i}^2 \hb_i^2}
\eean
We need to lower bound the last denominator. Recall that there is a $t$ s.t.\ $\hb_i=b_i-t$. Let $f(t)=\sum_i w_{i}^2 (b_i-t)^2$. Then $f$ is a quadratic in $t$ with positive leading coefficient, and $\partial f / \partial t= 2t\sum_i w_{i}^2-2\sum_iw_{i}^2b_i=2t\sum_i w_{i}^2$; so $f$ achieves its global minimum at $t=0$. Consequently,
$\frac{\sum_i w_{i}^2 b_i^2}{\sum_i w_{i}^2 \hb_i^2} \leq 1$
and therefore
\[ \lambda_{\uparrow 2}(\tL) \leq \nu \lambda_{\uparrow 2}(L) \]
proving the bound in the Lemma.

Turning to optimality of the Lemma: an example achieving this must focus the ``$w$'' weight away from the ``$b$'' weight, so that large jumps in $b$ occur only across weakly-weighted edges. This is achieved in the example of a chain $W$ of three edges in which the middle edge has weight $1$ and the outside edges weight $x$. One may calculate that $\lambda_{\uparrow 2} (\LAP(W))=\frac{1}{1+x}$. Now consider $\tW$ in which the outside edges have weight $x/\nu$. Then $\lambda_{\uparrow 2} (\LAP(\tW))/\lambda_{\uparrow 2} (\LAP(W)) =\nu \frac{1+x}{\nu+x}$. Fixing any $\nu$ and considering the limit of large $x$ we see that the supremum of this ratio is $\nu$.
\end{proof}
In order to apply the lemma we need to bound $\nu$; this will depend on
two basic parameters. These are $\gamma$, a kind of  measure of the complexity of the utility functions; and
$\tpsi$, which measures the disparity in the equilibrium prices across the network.
The first part of the comparison theorem assumes knowledge only of the underlying network (encapsulated in $U$), and of the numbers $\gamma$ and $\tpsi$.
The second part of the comparison theorem assumes that besides the  network and $\gamma$, we know also the equilibrium prices. This is reasonable if one is studying a functioning market near equilibrium. Let $\RG$ denote the weighted adjacency matrix with weights $\RG_{ij}=\sqrt{r_i r_j}$ on edges $i \sim j$ of the network (and $0$ elsewhere).
\begin{theorem}[Market comparison bounds] \label{MarketVsUnwghtd} \mbox{} \\
\begin{enumerate} \item \label{PartUnw} Bound through the unweighted Laplacian:
\bean
q(\delta, \min\{ \tpsi \gamma^{2+\delta} \lambda_{\uparrow 2}(\LAP(U)),1+\frac{1}{2\delta},1+\frac{1}{n-1}\})  &
\leq &
 \lambda_{\downarrow 2}(D) \\
& \leq &
\max\{q(\delta, \lambda_{\uparrow 2}(\LAP(U))/(\tpsi \gamma^{2+\delta})),-2\}. \eean
\item \label{PartEq} Bound through the equilibrium prices Laplacian:
\bean q(\delta, \min\{\gamma^{1+\delta} \lambda_{\uparrow 2}(\LAP(\RG)),1+\frac{1}{2\delta},1+\frac{1}{n-1}\})
& \leq &
 \lambda_{\downarrow 2}(D) \\
& \leq &
\max\{q(\delta, \lambda_{\uparrow 2}(\LAP(\RG))/\gamma^{1+\delta}),-2\}. \eean
\end{enumerate}
\end{theorem}

The complicated form of these bounds is misleading; in many and perhaps most cases of interest, the $\min$ or $\max$ is achieved by the term containing $\lambda_{\uparrow 2}$, hence the slack is fully captured by the factor of
$\tpsi \gamma^{2+\delta}$ or
$\gamma^{1+\delta}$
respectively. Exceptions are briefly discussed after Lemma~\ref{any-nu}.

Finally, we remind that the damping rate of a market is important not for what it says about (nonexistent) markets that are in true isolation; but rather for what it says about the predictive power of the equilibrium model, to a market that is buffeted by external noise. This is taken up in Section~\ref{noise}, where we show that in a certain noise model, and for any family of markets in which the prices are within a bounded range:
\begin{theorem}[Steady state distribution] \label{steady}
In steady state the prices are distributed according to a multivariate normal distribution, with largest directional variance proportional to $-1/\lambda_{\downarrow 2}(D)$. \end{theorem}

\section{Equilibrium: existence, uniqueness and detailed balance} \label{seq:eq}

The equilibrium equations Eq.~\eqref{demand-ruleG2} are homogeneous of degree $1$ in $r$, so any scalar multiple of an equilibrium vector $r$ is also an equilibrium vector. Subsequently when we discuss uniqueness, ``up to scaling'' is implied. Due to connectedness of the market, no price can be $0$ at an equilibrium. The existence of an equilibrium in our setting is a corollary of the theorem of Arrow and Debreu~\cite{ArrowD54} and McKenzie~\cite{McKenzie54} (improving on an earlier argument of Wald, see~\cite{Hildenbrand98,DW15}). In fact, as the utility functions are strongly concave and twice continuously differentiable, the equilibrium is unique (up to price scaling). We, however, require special properties of the equilibrium, and so need an existence proof which establishes these properties. In the course of providing this we incidentally give a self-contained proof of existence and uniqueness.

Adopting a term from probability theory, we say that a market is in detailed balance at prices $r$ if for every $i,j$, the payments from $i$ to $j$ equal those from $j$ to $i$. The payment from $i$ to $j$ is $d_{ij} r_j$ so the detailed balance conditions are
$
 \frac{ r_i C_{ij}}{r_j^{\delta}R_i}
=  \frac{ r_j C_{ji}}{r_i^{\delta}R_j}
$, or
\be
 \frac{r_i^{1+\delta} C_{ij}}{R_i}
=  \frac{ r_j^{1+\delta} C_{ji}}{R_j} \label{r-soln}
\ee

\begin{theorem} \label{exun} A market $C$ has unique equilibrium prices $r$,
and this market is in detailed balance at prices $r$. \end{theorem}
We remind that this is in the regime $\delta>0$, among the other assumptions detailed in Section~\ref{model}.
\begin{proof}
We begin with uniqueness.
\begin{lemma} There can be at most one equilibrium vector. \label{un} \end{lemma} \begin{proof}
Suppose there are two vectors $r,r'$ solving the equilibrium equations Eq.~\eqref{demand-ruleG2},
with $j$ a vertex minimizing $r'_j/r_j$ and having a neighbor $i$ (that is, an $i$ s.t.\ $C_{ij}>0$) which does not minimize this ratio. Rescale $r'$ so $r'_j=r_j$, $r'_i\geq r_i$ for all $i$, and $r'_i > r_i$ for some neighbor $i$ of $j$. Observe that since $\delta>0$, the quantity $r_i/R_i$ is a nondecreasing function of the price vector $r$, and moreover strictly increasing in $r_i$ and in any $r_k$ for $k$ a neighbor of $i$.
 Then applying Eq.~\eqref{demand-ruleG2} in numerator and denominator:
\[ 1= \frac{r'^{1+\delta}_j}{r_j^{1+\delta}}=
\frac{\sum_{i } \frac{ r'_i C_{ij}}{R'_i}}
{\sum_{i } \frac{ r_i C_{ij}}{R_i} }
> \frac{\sum_{i } \frac{r_i C_{ij}}{R_i}}
{\sum_{i } \frac{r_i C_{ij}}{R_i} }
=1
\]
a contradiction. \end{proof}

\begin{lemma} There exist prices $r$ satisfying the detailed balance conditions Eq.~\eqref{r-soln}.
\label{det-balG} \end{lemma} \begin{proof} 
Define for every two vertices $i,j$ the value 
\[   \psi_{i,j}=
 \prod_{\ell=1}^k \frac{C_{i_{\ell-1}i_\ell}}{C_{i_\ell i_{\ell-1}}} \; , \]
where $i=i_0,i_1,\ldots,i_k=j$ is a path in the graph; 
we claim this is well defined. Consider another path $i=i_0,i'_1,\ldots,i'_{k'}=j$ and form the cycle $i_0,i_1,\ldots,i_{k-1},j,i'_{k'-1},\ldots,i'_1,i_0$. The claim follows by the circulation-free property of $c$. Consequently we can fix $i_0$ to be a vertex such that $\psi_{i_0,j}\geq 1$ for all $j$, and define 
\be \psi_j=\psi_{i_0,j}
 \label{psi-defn} \ee
(with $\psi_{i_0}=1$). For future reference, observe that $\psi$ satisfies for any edge $i\sim j$ the identity
\be \psi_i C_{ij}=\psi_jC_{ji}\label{psi1} \ee
and that Eq.~\eqref{r-soln} yields by telescoping product another expression for $\psi$:
\be \psi_j=\frac{ r_j^{1+\delta} R_{i_0} }{ r_{i_0}^{1+\delta} R_j} .
\label{psi-tel} \ee

For use now and below we make several definitions: \begin{enumerate} 
\item $\tpsi=\max_j \psi_j = (\max_j \psi_j)/(\min_j \psi_j)$.
This is a \textit{global} measure of the disparity of prices in the network, and may of course be interpreted as a measure of the disparity of the desirability of the various goods.

\item $\gamma=\max_i \sum_j C_{ij}$. This is a \textit{local} measure of the diversity within each utility function. (Recall that each nonzero $C_{ij}$ is at least $1$.)

\item For $p$ a price vector, $p_{\max}=\max_j p_j$, $p_{\min}=\min_j p_j$, and $\tilde{p}=p_{\max}/p_{\min}$.
Let $|\cdot|$ denote geometric mean, so $|p|=\prod_1^n p_i^{1/n}$.

\end{enumerate}

Continuing with the proof, let $K=\{p: |p|=1, \tilde{p} \leq \gamma \tpsi \}$.
Let $f^0:K \to \rset^n$, \[ f^0_j(p)=|p|^{\frac{1+2\delta}{1+\delta}} \left(\psi_j \Pp_j(p)\right)^{\frac{1}{1+\delta}}.\]
Let $f:K \to \rset^n$, $f_j(p)=f^0_j(p)/|f^0(p)|$.
By Eq.~\eqref{psi1}, a fixed point of $f$ is a solution of Eq.~\eqref{r-soln}.

We now show that $f$ maps $K$ into $K$. By construction, $|f(p)|=1$; what we have to show is that $\widetilde{f(p)} \leq \gamma \tpsi$. This is equivalent to showing that $\widetilde{f^0(p)} \leq \gamma \tpsi$. We have
\bean f^0_j(p) &\leq& |p|^{\frac{1+2\delta}{1+\delta}} \left(\frac{\psi_{\max} \gamma }{ p_{\min}^\delta}\right)^{\frac{1}{1+\delta}}  \\
 f^0_j(p) &\geq& |p|^{\frac{1+2\delta}{1+\delta}} \left(\frac{\psi_{\min}}{ p_{\max}^\delta}\right)^{\frac{1}{1+\delta}} \eean
 so for $p \in K$,
 \[ \widetilde{f^0(p)} \leq \left( \tpsi \gamma \tilde{p}^\delta \right)^{\frac{1}{1+\delta}}
 \leq  \left( \tpsi \gamma (\gamma \tpsi)^\delta \right)^{\frac{1}{1+\delta}} = \gamma \tpsi.
 \]
Thus in fact $f:K \to K$. Since $K$ is compact and convex and $f$ is continuous on $K$,  the Brouwer fixed point theorem ensures $f$ has a fixed point in $K$.
\end{proof}

Finally, we compute, using the demand functions Eq.~\eqref{demandG},
 the total demand at $j$ for detailed-balance prices $r$

\be \sum_i d_{ij} =\frac{1}{r_j^{1+\delta}} \sum_i \frac{r_i C_{ij}}{R_i}
= \frac{1}{r_j^{1+\delta}} \sum_i \frac{r_j^{1+\delta} C_{ji}}{r_i^\delta R_j}
= \frac{1}{R_j} \sum_i \frac{C_{ji}}{r_i^\delta} = 1
 \ee
 where in the second equality we have applied  the detailed balance conditions Eq.~\eqref{r-soln}.

This shows that prices $r$ satisfying detailed balance necessarily satisfy the equilibrium conditions Eq.~\eqref{demand-ruleG2}.

By Lemmas~\ref{un},~\ref{det-balG} the Arrow-Debreu market possesses a solution which is unique and which moreover satisfies detailed balance.
\end{proof}

The above arguments imply in particular that in the Arrow-Debreu equilibrium, any two prices are within a factor of $\gamma \tpsi$.

\section{Quadratic expansion of the dynamics in terms of local interactions} \label{sec:dyn}
We now proceed to calculate $D$, the kernel of the dynamics, given by Eq.~\eqref{alphaD}:
\[ \ket{\dot{\alpha}_j} = D_{ji} \ket{\alpha_i} .\]
In order to state the outcome of this calculation it is necessary to define three matrices:
First, the (symmetric) matrix $\ell$ with entries
 \be \ell_{ji}= \sqrt{\frac{ C_{ij}C_{ji}  } {R_i r_i^{\delta}R_j r_j^{\delta}} }. \label{ell-defn} \ee
and its ``complement'' $\Lap=I-\ell$.

Next, the diagonal matrix $B$ with entries
 \be B_{jj} =\frac{r_j^{1+\delta/2}}{R_j^{1/2}\psi_j^{1/2}}.
 \label{rho-defn} \ee
Key to our work is that $\Lap$ contains all information necessary to express the system dynamics; this is encapsulated in Proposition~\ref{L-expansion}, which as we recall states that 
$ B D B^{-1}=q(\delta,\Lap)$, where  $\q(\delta,\lambda)=-(1+2\delta)\lambda+\delta \lambda^2$,
as given in Eq.~\eqref{q-def}. The rest of this section is devoted to proving Proposition~\ref{L-expansion}.

 \begin{proof}
The starting point for the calculation is Eq.~\eqref{demandG}. We consider the diagonal and off-diagonal entries of $D$ separately.

\begin{lemma} The entries $D_{jj}$ are:
\bean D_{jj} &=& 
 \frac{C_{jj}}{R_jr_j^\delta }
- (1+\delta)
+ \sum_i \frac{\delta C_{ij} C_{ji}}{R_i r_i^{\delta}R_jr_j^{\delta}}
\eean
And consequently
\bean D_{jj} &=& -1-\delta+\ell_{jj}+\delta\sum_i\ell_{ji}^2.
\eean
\end{lemma}
\begin{proof}

\bean D_{jj} &=& \sum_i \left. \left(\frac{\partial d_{ij}}{\partial \alpha_j} \right)\right|_{\ket{0}}
\\&=& \sum_i \left. \frac{\partial}{\partial \alpha_j} \right|_{\ket{0}} \frac{p_i C_{ij}}{p_j^{1+\delta} \sum_k C_{ik}/p_k^{\delta}}
\\&=&
 \left. \frac{\partial}{\partial \alpha_j} \right|_{\ket{0}} \left[ \frac{p_j C_{jj}}{p_j^{1+\delta} \sum_k C_{jk}/p_k^{\delta}}
+  \sum_{i \neq j} \frac{p_i C_{ij}}{p_j^{1+\delta} \sum_k C_{ik}/p_k^{\delta}}  \right]
\\
&=& \left. \frac{\partial}{\partial \alpha_j} \right|_{\ket{0}} \left[
\frac{e^{-\delta \alpha_j} C_{jj}}{ C_{jj}e^{-\delta \alpha_j} + r_j^{\delta} \sum_{k \neq j} C_{jk} /r_k^{\delta}}
+  \sum_{i \neq j} \frac{r_i C_{ij}}
{r_je^{\alpha_j} C_{ij}  + r_j^{1+\delta}e^{(1+\delta)\alpha_j}\sum_{k\neq j} C_{ik}/r_k^{\delta}}  \right]
\suppress{
\\ &=& \frac{-\delta (r_j^{\delta} \sum_{k} C_{jk} /r_k^{\delta}) C_{jj} +\delta C_{jj}^2 }
{ (r_j^{\delta} \sum_{k} C_{jk} /r_k^{\delta})^2}
+  \sum_{i \neq j} \frac{-r_i C_{ij}(r_jC_{ij}+(1+\delta)r_j^{1+\delta}\sum_{k\neq j} C_{ik}/r_k^{\delta})}
{( r_j^{1+\delta}\sum_{k} C_{ik}/r_k^{\delta})^2}
}
\\ &=& \frac{-\delta r_j^{\delta} R_j C_{jj}
+\delta C_{jj}^2 }
{ r_j^{2\delta} R_j^2}
+  \sum_{i \neq j} \frac{-r_i C_{ij}(-\delta r_jC_{ij}+(1+\delta)r_j^{1+\delta}R_i)}
{r_j^{2+2\delta}R_i^2}
\suppress{
\\&=& \frac{\delta C_{jj}^2 -\delta r_j^{\delta} R_j C_{jj}}
{ r_j^{2\delta} R_j^2}
+ \sum_{i \neq j} r_i \frac{\delta r_jC_{ij}^2 -(1+\delta) C_{ij}r_j^{1+\delta}R_i}
{r_j^{2+2\delta}R_i^2}
}
\\ &=& \frac{\delta C_{jj}^2}
{ r_j^{2\delta} R_j^2}
-\frac{\delta C_{jj}}{ r_j^{\delta} R_j}
+ \sum_{i \neq j} \frac{ 1}{R_i} r_i \left( \frac{\delta r_jC_{ij}^2}
{r_j^{2+2\delta}R_i}
-
 \frac{(1+\delta) C_{ij}}
{r_j^{1+\delta}}  \right)
\eean

We replace the factor $\frac{1}{R_i}$ inside the last summation by a term depending on $R_j$ using Eq.~\eqref{r-soln} which gives
$R_i= \frac{ r_i^{1+\delta} C_{ij}R_j}
{r_j^{1+\delta} C_{ji}}$.
So
\bean D_{jj} &=&  \frac{\delta C_{jj}^2}
{ r_j^{2\delta} R_j^2}
-\frac{\delta C_{jj}}{ r_j^{\delta} R_j}
+  \sum_{i \neq j} r_i
\left(\frac{\delta r_jC_{ij}^2 r_j^{1+\delta} C_{ji}}
{r_j^{2+2\delta}R_i r_i^{1+\delta} C_{ij}R_j}
-
 \frac{(1+\delta) C_{ij} r_j^{1+\delta} C_{ji}}
{r_j^{1+\delta} r_i^{1+\delta} C_{ij}R_j}
  \right)
\\ &=&
 \frac{\delta C_{jj}^2}
{ r_j^{2\delta} R_j^2}
-\frac{\delta C_{jj}}{ r_j^{\delta} R_j}
+   \sum_{i \neq j}
\left(\frac{\delta C_{ij} C_{ji}}
{r_j^{\delta}R_i r_i^{\delta}R_j}
-
 \frac{(1+\delta) C_{ji}}
{R_j r_i^{\delta} }
  \right)
\\&=& \frac{C_{jj}}{R_jr_j^\delta }
- \sum_i \frac{(1+\delta) C_{ji}}{R_j r_i^{\delta} }
+ \sum_i \frac{\delta C_{ij} C_{ji}}{R_i r_i^{\delta}R_jr_j^{\delta}}
\\&=& \frac{C_{jj}}{R_jr_j^\delta }
- (1+\delta)
+ \sum_i \frac{\delta C_{ij} C_{ji}}{R_i r_i^{\delta}R_jr_j^{\delta}}
\eean
\end{proof}
\begin{lemma} The entries $D_{jk}$, $k \neq j$ are:
\bean D_{jk} &=& \frac{ r_k C_{kj}}
{r_j^{1+\delta} R_k}
+ \frac{\delta  r_k^{1-\delta} C_{kk} C_{kj}}
{r_j^{1+\delta} R_k^2}
+ \sum_{i \neq k} \frac{\delta r_i C_{ij}C_{ik} }
{r_k^\delta r_j^{1+\delta}R_i^2}
\eean
\end{lemma}
\begin{proof}
\bean D_{jk} &=& \sum_i \left. \left(\frac{\partial d_{ij}}{\partial \alpha_k} \right)\right|_{\ket{0}}
\\&=& \sum_i \left. \frac{\partial}{\partial \alpha_k} \right|_{\ket{0}} \frac{p_i C_{ij}}{p_j^{1+\delta} \sum_h C_{ih}/p_h^{\delta}}
\\&=&
 \left. \frac{\partial}{\partial \alpha_k} \right|_{\ket{0}} \left[ \frac{p_k C_{kj}}{p_j^{1+\delta} \sum_h C_{kh}/p_h^{\delta}}
+  \sum_{i \neq k} \frac{p_i C_{ij}}{p_j^{1+\delta} \sum_h C_{ih}/p_h^{\delta}}  \right]
\\ &=&
 \left. \frac{\partial}{\partial \alpha_k} \right|_{\ket{0}} \left[ \frac{r_k e^{\alpha_k} C_{kj}}{r_j^{1+\delta}(C_{kk}e^{-\delta \alpha_k} /r_k^\delta+ \sum_{h \neq k} C_{kh}/r_h^{\delta})}
+ \sum_{i \neq k} \frac{r_i C_{ij}}{r_j^{1+\delta}(C_{ik}e^{-\delta \alpha_k} /r_k^\delta+ \sum_{h\neq k} C_{ih}/r_h^{\delta})}  \right]
\suppress{
\\ &=&
 \frac{((r_j^{1+\delta} \sum_{h} C_{kh}/r_h^{\delta})
+\delta r_j^{1+\delta} C_{kk}/r_k^\delta)r_k C_{kj}}
{(r_j^{1+\delta} \sum_{h} C_{kh}/r_h^{\delta})^2}
+ \sum_{i \neq k} \frac{\delta r_i C_{ij}r_j^{1+\delta}C_{ik}/r_k^\delta }
{(r_j^{1+\delta}(\sum_{h} C_{ih}/r_h^{\delta}))^2}
}
\\ &=& \frac{(r_j^{1+\delta} R_k
+\delta r_j^{1+\delta} C_{kk}/r_k^\delta)r_k C_{kj}}
{r_j^{2+2\delta} R_k^2}
+ \sum_{i \neq k} \frac{\delta r_i C_{ij}r_j^{1+\delta}C_{ik}/r_k^\delta }
{r_j^{2+2\delta}R_i^2}
\\ &=& \frac{ r_k C_{kj}}
{r_j^{1+\delta} R_k}
+ \frac{\delta  r_k^{1-\delta} C_{kk} C_{kj}}
{r_j^{1+\delta} R_k^2}
+ \sum_{i \neq k} \frac{\delta r_i C_{ij}C_{ik} }
{r_k^\delta r_j^{1+\delta}R_i^2}
\eean
\end{proof}
The dynamics matrix $D$ is not symmetric but we can symmetrize it by
the change of basis $\ket{\beta}=B \ket{\alpha}$. So we study the dynamics on $\beta$ obtained by similarity transform:
\be \ket{\dot{\beta}}=B D B^{-1} \ket{\beta} \label{betadyn} \ee

Now
\bean (B D B^{-1})_{jk} &=&
 \sqrt{\frac{R_k\psi_k r_j^{2+\delta}}{R_j\psi_j r_k^{2+\delta}}}
\left( \frac{r_k C_{kj}}
{r_j^{1+\delta} R_k}
+ \frac{\delta r_k^{1-\delta} C_{kk} C_{kj}}
{r_j^{1+\delta} R_k^2}
+ \sum_{i \neq k} \frac{1}{R_i} \frac{\delta r_i C_{ij}C_{ik} }
{r_k^\delta r_j^{1+\delta}R_i}
\right)
\eean
We replace the factor $\frac{1}{R_i}$ inside the last summation by a term depending on $R_k$ using Eq.~\eqref{r-soln}, yielding
\bean (B D B^{-1})_{jk} &=&
\suppress{
  \sqrt{\frac{R_k\psi_k r_j^{2+\delta}}{R_j\psi_j r_k^{2+\delta}}}
\left( \frac{ r_k C_{kj}}
{r_j^{1+\delta} R_k}
+ \frac{\delta r_k^{1-\delta} C_{kk} C_{kj}}
{r_j^{1+\delta} R_k^2}
+ \sum_{i \neq k} \frac{\delta s_ir_i C_{ij}C_{ik} r_k^{1+\delta} C_{ki}}
{r_k^\delta r_j^{1+\delta}R_i s_i r_i^{1+\delta} C_{ik}R_k}
\right)
\\ &=&
}
  \sqrt{\frac{R_k\psi_k r_j^{2+\delta}}{R_j\psi_j r_k^{2+\delta}}}
\left( \frac{ r_k C_{kj}}
{r_j^{1+\delta} R_k}
+ \frac{\delta r_k^{1-\delta} C_{kk} C_{kj}}
{r_j^{1+\delta} R_k^2}
+ \sum_{i \neq k} \frac{\delta C_{ij} r_k C_{ki}}
{r_j^{1+\delta}R_i r_i^{\delta} R_k} \right)
\\ &=&
  \sqrt{\frac{R_k\psi_k r_j^{2+\delta}}{R_j\psi_j r_k^{2+\delta}}}
 \frac{ r_k C_{kj}}
{r_j^{1+\delta} R_k}
+ \sqrt{\frac{R_k\psi_k r_j^{2+\delta}}{R_j\psi_j r_k^{2+\delta}}} \frac{\delta r_k^{1-\delta} C_{kk} C_{kj}}
{r_j^{1+\delta} R_k^2}
+ \sum_{i \neq k}
 \sqrt{\frac{\psi_k}{R_jR_k \psi_j r_k^{\delta}r_j^{\delta}}}
\frac{\delta C_{ij}  C_{ki}}
{R_i r_i^{\delta}}
\eean
Apply the identity Eq.~\eqref{psi1} to the edges $j \sim k$ in the first two terms, and $i\sim j$ and $i\sim k$ in the third term.
\bean (B D B^{-1})_{jk} &=&
  \sqrt{\frac{ C_{jk} C_{kj}}{R_jr_j^{\delta}R_k r_k^{\delta} }  }
+ \sqrt{\frac{\delta C_{kj}C_{jk}}{R_j r_j^\delta R_k r_k^\delta}}\sqrt{\frac{\delta C_{kk}^2}{R_k^2 r_k^{2\delta}}}
+ \sum_{i \neq k}
 \sqrt{\frac{\delta C_{ij}C_{ji}  }
{R_i r_i^{\delta}R_j r_j^{\delta}} }
\sqrt{\frac{\delta C_{ki}C_{ik}}{R_i r_i^\delta R_kr_k^{\delta}}}
\\ &=&
  \sqrt{\frac{ C_{jk} C_{kj}}{R_jr_j^{\delta}R_k r_k^{\delta} }  }
+ \sum_{i}
 \sqrt{\frac{\delta C_{ij}C_{ji}  }
{R_i r_i^{\delta}R_j r_j^{\delta}} }
\sqrt{\frac{\delta C_{ki}C_{ik}}{R_i r_i^\delta R_kr_k^{\delta}}}
\\&=&  \ell_{jk} + \delta \sum_i \ell_{ji}\ell_{ik}
\eean
\subsection*{The quadratic expansion formula}
Collecting the calculations of $D_{jj}$ and $D_{jk}$, we have:
\bean (B D B^{-1})_{jk} &=& \begin{cases}
-1-\delta+\ell_{jj}+\delta\sum_i\ell_{ji}^2 \quad & \text{if } j=k \\
\ell_{jk} + \delta \sum_i \ell_{ji}\ell_{ik} & \text{if } j \neq k \end{cases} \eean

Finally, we can complete the proof of Proposition~\ref{L-expansion}: $-(1+2\delta) \Lap+ \delta \Lap^2 =-(1+\delta)I + \ell + \delta \ell^2 = B D B^{-1}$.
\end{proof}

\section{The dynamics in terms of the market Laplacian $\Lap$}
In general terms Proposition~\ref{L-expansion} gives what we have been seeking: an expression for the system dynamics, in terms of a symmetric matrix whose off-diagonal entries are supported on the edges of the network. In order to make this more quantitative we start by fulfilling our earlier promise and showing how $ \Lap$ may be represented as the Laplacian $\Lap=\LAP(W)$ of a suitable edge-weighting $W$ of the network. As described earlier, this requires that
$ \Lap = I - w^{-1} W w^{-1} $, equivalently
\[ \ell= w^{-1} W w^{-1} \]
where $w=\sigma(W)$. It will simplify notation to write $w_i=w_{ii}$.

Now equivalently, with $\ket{w_i}=w\ket{1}$ denoting the vector containing the diagonal entries of $w$,
\[ w_i^2=\sum_j  w_i \ell_{ij} w_j \]
\[ \ket{w_i}= \ell \ket{w_i}. \]
So we wish to identify the kernel of $I-\ell=\Lap$. We show that this kernel equals
$B \ket{1}$. First, in order to verify that $B \ket{1}$ is in the kernel, we obtain from
Eqs.~\eqref{rho-defn} and \eqref{psi-tel}
that for all $k$,
\be B_{kk} =  \sqrt{ \frac{r_k r_{i_0}^{1+\delta}}
{R_{i_0} } } \label{Bagain} \ee
(where $i_0$ is any fixed vertex as defined in Lemma~\ref{det-balG}). Now
\[ (\Lap B \ket{1})_j = \sqrt{\frac{ r_{i_0}^{1+\delta}}{R_{i_0} }} \sum_k \sqrt{r_k } \Lap_{jk}
= \sqrt{\frac{ r_{i_0}^{1+\delta}}{R_{i_0} }}  \left(\sqrt{r_j } - \sum_{k} \sqrt{r_k } \ell_{jk} \right)
\]
If we combine Eqs~\eqref{r-soln} and \eqref{ell-defn} we see that $\ell_{jk}=\frac{C_{jk} \sqrt{r_j}}{R_j\sqrt{r_k^{1+2\delta}}}$. Substituting, we have
\bean (\Lap B \ket{1})_j &=&  \sqrt{\frac{r_{i_0}^{1+\delta}}{R_{i_0} }}  \sqrt{r_j} (1-\frac{1}{R_j} \sum_k \frac{C_{jk}}{r_k^\delta}) = 0 \eean
as required.

Thus we have that $w_i=B_{ii}$, which is to say $w=B$; and that the matrix $W$ formed by
\be W=B\ell B \label{WfromB} \ee
is the weighted adjacency matrix corresponding to the Laplacian $\Lap$. Since all $B_{ii}$ are positive and the nonzero off-diagonal entries of $\Lap$ form a connected graph, the kernel of $\Lap$ is of rank $1$, and all remaining eigenvalues of $\Lap$ are positive. As is well known, all are bounded above by $2$.

\subsection{The corresponding random walk}
Although not strictly required for our work, we pause to describe the random walk which corresponds to the undirected edge-weighted graph $W$. Acting as r.w.\ on column vectors, it is the stochastic matrix obtained by rescaling each column $i$ of $W$ by $w_i^2$, equivalently $B_{ii}^2$, namely
$WB^{-2}$. This in turn can be rewritten as
$B \ell B^{-1}$. One may readily verify that $\bra{1}$ is invariant:
 \[ \bra{1}B\ell B^{-1} = \bra{1} W B^{-2}=\bra{B_{ii}^2}B^{-2} = \bra{1} \]
and the corresponding right eigenvector, the stationary distribution, is $\ket{B_{ii}^2}$:
\[ B\ell B^{-1}\ket{B_{ii}^2} = W B^{-2}\ket{B_{ii}^2} = W \ket{1} = \ket{B_{ii}^2} .\]
Our detailed balance condition agrees here with the detailed balance condition of the random walk, namely that for any $i \sim j$, the frequency of transitions in each direction across the edge are equal, which we can verify here by:
\[ (W B^{-2})_{ji} B_{ii}^2 = W_{ji} = W_{ij} =
(W B^{-2})_{ij} B_{jj}^2\]

\subsection{Proof of Theorem~\ref{precise}: bounds on $\lambda_{\downarrow 2}(B D B^{-1})$ in terms of $\lambda_{\uparrow 2}(\Lap)$} \label{sec:genbds}
Ultimately what we are interested in is
 the second-largest eigenvalue of $B D B^{-1}$ (the largest being $0$ and corresponding to the ray of equilibrium states), denoted $\lambda_{\downarrow 2}(B D B^{-1})$, which we called earlier the damping rate of the market, because it scales as $-1/T$ for $T$ the half-life of a perturbation from equilibrium. (It is convenient that the eigenvalues are real, but of course the familiar property of linear systems that we are exploiting and quantifying here is that the eigenvalues are in the open left-half plane.)

Due to Proposition~\ref{L-expansion}
every eigenvalue $\lambda$ of $\Lap$ maps to an eigenvalue
$\q(\delta,\lambda)=-(1+2\delta)\lambda+\delta \lambda^2$ of $B D B^{-1}$.
The mapping $\q$ is monotone decreasing in $\lambda$ throughout $[0,1+\frac{1}{2\delta}]$; if
$1+\frac{1}{2\delta}<2$ it then rises, symmetrically, to $-2$ at $\lambda=2$, and it also equals $-2$ at $\frac{1}{\delta}$. (See Figure.) Thus, sufficient conditions that
$\lambda_{\downarrow 2}(B D B^{-1})=\q(\delta,\lambda_{\uparrow 2}(\Lap))$ include that (a)
The spectrum of $\Lap$ is contained in $[0,1+\frac{1}{2\delta}]$, or (b) $\lambda_{\uparrow 2}(\Lap) \leq
\frac{1}{\delta}$.

Clause (a) will occur if $W$ is  ``far from bipartite'', in particular if there is sufficient local consumption of goods (i.e., the coefficients $C_{jj}$ are large enough). Clause (b) is in particular guaranteed if $\delta\leq 1/2$.

Even outside these favorable cases, note that for any $\lambda>1+\frac{1}{2\delta}$, $q(\delta,\lambda)\leq -2$. Consequently in all cases:
\be
q(\delta,\lambda_{\uparrow 2}(\Lap ))
\leq
 \lambda_{\downarrow 2}(B D B^{-1})
\leq \max\{
q(\delta,\lambda_{\uparrow 2}(\Lap )), -2\}
\label{bdb-lap} \ee

This proves Theorem~\ref{precise}.
\begin{center}
\includegraphics[height=45mm]{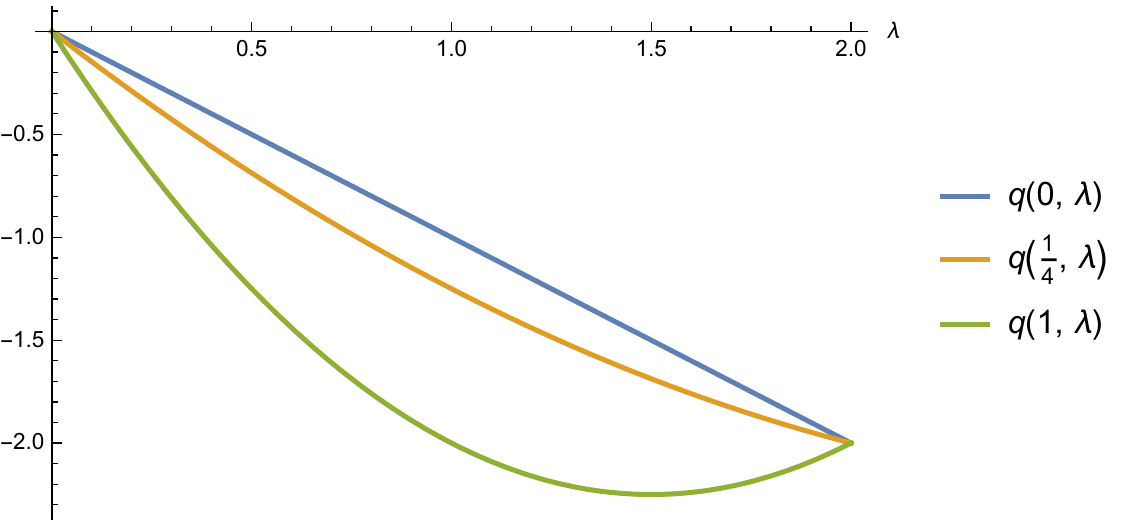}
\end{center}

\subsection{Proof of Theorem~\ref{thm:unif}: the ``uniform'' special case}
Earlier we referred to the special case that all $C_{ij} \in \{0,1\}$ and there is a $\Delta>1$ s.t.\ for all $i$, $|\{j: C_{ij}=1\}|=|\{j: C_{ji}=1\}|=\Delta$. In this case all $r_i=1$, $B$ is the identity, $\ell=W$, and $W_{ij}=0$ unless $i \sim j$ in which case $W_{ij}=1/\Delta$. Then $\Lap=\LAP(U)$. So Eq.~\eqref{bdb-lap} proves Theorem~\ref{thm:unif}.

\section{Comparison theorems: the market Laplacian $\Lap$ vs.\ the unweighted Laplacian $\LAP(U)$ and the equilibrium prices Laplacian $\LAP(\RG)$.} We have been working so far with a Laplacian $\Lap$ which carries very detailed information about the parameters $C$ of the market, and moreover, depends on those parameters in an indirect way (through the equilibrium prices).
Generally however, when studying a market, we may not know $C$ or $\Lap$; and things will certainly not be as simple as in Theorem~\ref{thm:unif}.
 One of our goals is to have (still two-sided) bounds on the convergence time which are weaker than Theorem~\ref{precise} but require less edge-by-edge information: only the connectivity structure of the market (i.e., the matrix $U$ of the unweighted network)---which is revealed from actual consumption---and, optionally, the equilibrium prices. The key to this is Lemma~\ref{lap:stab}, which will enable us in this Section to bound  $\lambda_{\downarrow 2}(B D B^{-1})$ in terms of the unweighted network Laplacian $\LAP(U)$ or the equilibrium prices Laplacian $\LAP(\RG)$.

In both results we pay an approximation factor polynomial in the ``utility function complexity'' $\gamma$; in the former result, we also pay a factor proportional to the ``price disparity'' $\tpsi$
(both parameters were defined in Section~\ref{seq:eq}). Both these dependences are necessary, as we discuss in the next Section.

First we exhibit that $\tpsi$ can be exponential in the network size even though $\gamma$ is constant.
It is therefore advantageous, in applying these results, to know the equilibrium prices.
\begin{example} Take $\delta=1$. Fix any $a>1$ and create a market among participants $1,\ldots,n$ arranged in a chain, as follows. $C_{ij}$ is nonzero only for $|i-j|\leq 1$. For such $i,j$, $C_{ij}=a^{j-i}$.
Up to some edge-effects, prices in this network are proportional to $a^{2i}$. Thus $\tpsi \in \Theta(a^{2n})$.
\end{example}
Before proving Theorem~\ref{MarketVsUnwghtd}
 we first establish two lemmas about the equilibrium prices:
 \begin{lemma} $\max_{i \sim j} \frac{r_i}{r_j} \leq \gamma $. \label{rratio} \end{lemma}
\begin{proof}
Let $i\sim j$ be such that $\mu= \frac{r_i}{r_j} =\max_{i' \sim j'} \frac{r_{i'}}{r_{j'}}$. Applying Eq.~\eqref{r-soln}:
\[ \mu^{1+\delta}=
\frac{r_i^{1+\delta}}{ r_j^{1+\delta}}
=\frac{C_{ji} R_i} {C_{ij} R_j}
= \frac{C_{ji} \sum_k C_{ik}/r_k^{\delta}} {C_{ij} \sum_\ell C_{j\ell}/r_\ell^\delta}
\leq \left(\frac{\mu^2 r_j}{r_i}\right)^\delta
\frac{C_{ji} \sum_k C_{ik}} {C_{ij} \sum_\ell C_{j\ell}}
= \mu^\delta \frac{C_{ji} \sum_k C_{ik}} {C_{ij} \sum_\ell C_{j\ell}}
\]
and applying $C_{ji}\leq \sum_\ell C_{j\ell}$, the lemma follows. \end{proof}

\begin{lemma} For any $i \sim j$, $\gamma^{-2-2\delta} \leq \frac{C_{ij}C_{ji}} {r_i^\delta R_i r_j^\delta R_j} \leq 1$.
That is,
$\gamma^{-1-\delta}
\leq \ell_{ij} \leq 1$. \label{lbd} \end{lemma}
\begin{proof}
The upper bound follows by dropping most terms in the denominator, leaving only $r_i^\delta \frac{C_{ij}} {r_j^\delta} r_j^\delta \frac{C_{ji}}{ r_i^\delta}$.

For the lower bound we apply Lemma~\ref{rratio} to get $R_i \leq (\gamma /r_i)^{\delta}
\sum_k C_{ik}$. Applying the same to $R_j$ yields
$ \frac{C_{ij}C_{ji}} {r_i^\delta R_i r_j^\delta R_j} \geq
 \frac{C_{ij}C_{ji}}{\gamma^{2\delta} (\sum_k C_{ik})(\sum_\ell C_{j \ell})} \geq \gamma^{-2-2\delta}
$. \end{proof}

As $D$ and $BDB^{-1}$ are cospectral, all bounds from now on are stated in terms of the latter which, being symmetric, is easier to work with. So we proceed with the proof of  Theorem~\ref{MarketVsUnwghtd}, replacing $\lambda_{\downarrow 2}(D)$
in it with  $\lambda_{\downarrow 2}(B D B^{-1})$.

\begin{proof}
The proof breaks into two lemmas. The first is a general bound on the damping rate of our market (with adjacency matrix $W$ and Laplacian $\Lap=\LAP(W)$), in terms of two features of any other adjacency matrix $W'$: the spectrum of its Laplacian, and  $\nu(W,W')$. The second lemma bounds $\nu(W,U)$ and $\nu(W,\RG)$.

\begin{lemma} \label{any-nu} Let $\nu=\nu(W,W')$. Then
\bean q(\delta, \min\{\nu \lambda_{\uparrow 2}(\LAP(W')),1+\frac{1}{2\delta},1+\frac{1}{n-1}\})
& \leq &
 \lambda_{\downarrow 2}(B D B^{-1}) \\
& \leq &
\max\{q(\delta, \lambda_{\uparrow 2}(\LAP(W'))/\nu),-2\}. \eean
\end{lemma} \begin{proof}
 For the first inequality in the Lemma, recall $
 \lambda_{\downarrow 2}(B D B^{-1}) \geq q(\delta,\lambda_{\uparrow 2}(\Lap))$ from Eq.~\eqref{bdb-lap};
 also note that
  $q(\delta,\lambda)$ is monotone decreasing in $\lambda$ until the global minimum at $1+\frac{1}{2\delta}$. We have two upper bounds on $  \lambda_{\uparrow 2}(\Lap)$:
$  \lambda_{\uparrow 2}(\Lap) \leq
 \nu \lambda_{\uparrow 2}(\LAP(W'))$ from Lemma~\ref{lap:stab}, and
  $  \lambda_{\uparrow 2}(\Lap) \leq 1+\frac{1}{n-1}$ because $\Tr(\Lap)\leq n$. Consequently
  $ \lambda_{\downarrow 2}(B D B^{-1}) \geq q(\delta, \min\{\nu \lambda_{\uparrow 2}(\LAP(W')),1+\frac{1}{2\delta},1+\frac{1}{n-1}\})$.

For the second inequality, recall
$ \lambda_{\downarrow 2}(B D B^{-1})
\leq \max\{
q(\delta,\lambda_{\uparrow 2}(\Lap )), -2\}$ from Eq.~\eqref{bdb-lap}. If $ \lambda_{\downarrow 2}(B D B^{-1}) > -2$ then necessarily
$q(\delta,\lambda_{\uparrow 2}(\Lap ))> -2$, and then we must have $\lambda_{\uparrow 2}(\Lap )<\min\{1/\delta,2\}$. This implies $q$ is monotone in the interval $[0,\lambda_{\uparrow 2}(\Lap )]$; then applying
$\lambda_{\uparrow 2}(\LAP(W'))/\nu \leq
\lambda_{\uparrow 2}(\Lap)$
from Lemma~\ref{lap:stab},
we find
$q(\delta,\lambda_{\uparrow 2}(\Lap )) \leq q(\delta, \lambda_{\uparrow 2}(\LAP(W'))/\nu)$.
\end{proof}
\begin{remark} An equivalent form of the upper bound in Lemma~\ref{any-nu} is
\[ \max\{q(\delta, \lambda_{\uparrow 2}(\LAP(W'))/\nu),-2\} =
 q(\delta,\min\{\lambda_{\uparrow 2}(\LAP(W'))/\nu, 1/\delta\}). \]
\end{remark}
\begin{remark} In most cases of interest the bounds in Lemma~\ref{any-nu} are determined by the $\lambda_{\uparrow 2}(\LAP(W'))$ term. For $\delta\leq 1/2$ $q$ is monotone in $[0,2]$, so this is guaranteed. Even outside this range ``most'' graphs will have $\lambda_{\uparrow 2}(\LAP(W'))$ small enough for this to hold (unless $\nu$ is very large; but then the comparison theorem is of course rather weak to begin with). Nevertheless it is worth pointing out an example, even with $\nu=1$, when the bound is not determined by $\lambda_{\uparrow 2}(\LAP(W'))$. Take the complete bipartite graph $K_{2,2}$. Its Laplacian spectrum is $\{0,1,1,2\}$. For $\delta>1/2$ the critical eigenvalue here is not $\lambda_{\uparrow 2}(\LAP(K_{2,2}))=1$, but $\lambda_{\uparrow 4}(\LAP(K_{2,2}))=2$, and correspondingly the damping rate is $-2$.
\end{remark}

\begin{lemma} $\nu(W,U) \leq \tpsi \gamma^{2+\delta}$ and
 $\nu(W,\RG)\leq \gamma^{1+\delta}$. \end{lemma} \begin{proof}
Consider the entries of the weighted adjacency matrix, $W=B\ell B$. Applying Lemma~\ref{lbd} and Eq.~\eqref{Bagain},
\be
\frac{ r_{i_0}^{1+\delta}}{R_{i_0}}
\sqrt{r_ir_j} \gamma^{-1-\delta}
\leq W_{ij} \leq
\frac{ r_{i_0}^{1+\delta}}{R_{i_0}}
\sqrt{r_ir_j}
\label{WvsEqbm} \ee
 Earlier (Section~\ref{seq:eq}) we bounded the variation in prices in terms of $\gamma \tpsi$, and so we have that if
$W_{ij}\neq 0$ then for any $i',j'$:
$W_{i'j'}/W_{ij} \leq \tpsi \gamma^{2+\delta}$.
This implies the first bound in the lemma.
The second bound in the lemma follows immediately from Eq.~\eqref{WvsEqbm}.\end{proof} This completes the proof of Theorem~\ref{MarketVsUnwghtd}. \end{proof}

\begin{remark}[on the necessity of the dependences on $\gamma$ and $\tpsi$ in Theorem~\ref{MarketVsUnwghtd}]
First, concerning $\gamma$: it is clear that the bounds in Theorem~\ref{MarketVsUnwghtd} must depend on $\gamma$ because very ``weak'' edges, those expressing little interest of participant $i$ in good $j$, are in the unweighted graph indistinguishable from any other edge; their presence must therefore weaken the quality of the result that we can obtain from knowing only the unweighted graph. Weak edges express themselves in our parameters by forcing $\gamma$ to be large.

Next, concerning $\tpsi$: the main difference in the strengths of Part~(\ref{PartUnw}) and Part~(\ref{PartEq}) is that in the latter we do not lose the factor of $\tpsi$ due to the disparity in prices.
 One might ask whether the factor of $\tpsi$ in the bound is an artifact of the analysis. The answer is that it is not: the dependence on $\tpsi$ in Part~(\ref{PartUnw}) is unavoidable. In markets with very unbalanced prices, even if $\gamma$ is bounded, the Laplacian $\LAP(U)$ of the unweighted graph can be an exponentially poor proxy for the actual market Laplacian $\Lap$. We now show this.

\textit{Example showing exponential gap in damping rate between $\Lap$ and $\LAP(U)$.} 
Our example uses $\delta=1$. Fix any $A>1$ and create a market among participants $-n,\ldots,n$ arranged in a chain, so that $C_{ij}>0$ if and only if $|i-j|\leq 1$. We will show how to set the coefficients $C_{ij}$ in a bounded range so that $r_i=A^{|i|}$. We describe the $C_{ij}$'s for $i\geq 0$; the construction will be symmetric about the origin.

Then Eq.~\eqref{r-soln} is, first at $0$, then at $0<i<n$, and then at $n$:
\bea
\frac{C_{01}}{1+2C_{01}A^{-1}}&=&\frac{A^2C_{10}}{C_{10}+A^{-1}+C_{12}A^{-2}} \label{0i} \\
\frac{C_{i,i+1}A^{2i}}{A^{1-i}C_{i,i-1}+A^{-i}+C_{i,i+1}A^{-i-1}} &=&
   \frac{A^{2i+2}C_{i+1,i}}{A^{-i}C_{i+1,i}+A^{-i-1}+C_{i+1,i+2}A^{-i-2}} \label{1in} \\
\frac{C_{n-1,n}A^{2n-2}}{A^{2-n}C_{n-1,n-2}+A^{-n+1}+C_{n-1,n}A^{-n}} &=&
   \frac{A^{2n}C_{n,n-1}}{A^{-n+1}C_{n,n-1}+A^{-n}} \label{in}
   \eea
Next specialize to taking all $C_{i,i+1}=A$. Then Eq.~\eqref{1in} becomes
\bean
\frac{A^{2i+1}}{A^{1-i}C_{i,i-1}+2A^{-i}} &=&
   \frac{A^{2i+2}C_{i+1,i}}{A^{-i}C_{i+1,i}+2A^{-i-1}}
   \eean
 It turns out that $C_{i+1,i}$ converges rapidly to $A^{-2}$ which we can see from writing $C_{i,i-1}=A^{-2}+\alpha_i$, $C_{i+1,i}=A^{-2}+\alpha_{i+1}$, and deriving the recurrence $\alpha_{i+1}=g(\alpha_i)$ where
\[ g(x):=\frac{-x}{A^2+2A} \]
maps the interval $(-1/A^2,1)$ into itself.

It remains only to show that the boundary conditions can be satisfied consistent with these choices. In Eq.~\eqref{0i}, which becomes $A/3=A^2C_{10}/(C_{10}+2/A)$, we have
$C_{10}=\frac{2}{3A^2-A}$, which lies in the interval $(-1/A^2,1)$ for any $A>1$. It remains only to see that there is a positive solution to Eq.~\eqref{in}, which becomes $A^{n-1}/(A^{3-n}+2A^{1-n})=A^nC_{n,n-1}/(A^{1-n}C_{n,n-1}+A^{-n})$ and is solved by $C_{n,n-1}=\frac{1}{A^4+2A^2-A}$ which is indeed bounded away from $0$.

Now that we have such a simple representation for the equilibrium prices, we can examine the weighted graph. Note that all nonzero entries of $\ell$ are within a constant factor (depending on $A$) of $1$.
From Eqs.~\eqref{Bagain} (with $i_0=0$) and~\eqref{WfromB} we see that $R_{i_0}=3$, $B_{kk}=\sqrt{A^{|k|}/3}$, and the weight of edge $i \sim i+1$ is within a constant factor of  $A^{|i|+1/2}$. Therefore, splitting this graph about the origin, we see that  its conductance is proportional to $A^{-n}$. From the Cheeger inequalities~\cite{AlonM85,Alon86,SinclairJ89} we can conclude that the algebraic connectivity, too, is exponentially small in $n$. The algebraic connectivity of the unweighted chain, by contrast, is proportional to $1/n^2$.

Thus, in this market, price equilibration is exponentially slower than that of a market that has the same connectivity structure but in which all goods have the same price. 
\end{remark}

\section{Markets subject to noise} \label{noise}
The obvious implication of the damping rate is that it tells how quickly an isolated market that has been disrupted, will converge
 back to equilibrium. But a more interesting implication concerns markets which are not isolated, but instead can be modeled as continually disturbed by noise. In this case, as we now show, the damping rate, rather than telling us something about the dynamics of the market, tells us instead about its steady state---which is no longer a single equilibrium point, but a probability distribution over the price space.

We do not attempt a ``realistic'' treatment of noise. We analyze a noise model chosen for analytic convenience and speculate that the qualitative predictions will hold for other light-tailed noise models.
The noise model is additive diffusion in the $\beta$ basis. That is, $\ket{\beta}$ undergoes a combination of the deterministic evolution derived in Eq.~\eqref{betadyn}, plus a (spherically symmetric) Brownian motion term. The description of the system at any time is a probability distribution over $\ket{\beta}$. The deterministic evolution term acts to contract the distribution toward the equilibrium ray of prices; while the diffusion term acts as an entropic force resisting over-contraction. The combined process diagonalizes in the basis of eigenvectors of $\Lap$, and therefore of $D$; so we may, for any eigenvector $v$ of $D$ with eigenvalue $\lambda$, replace
the deterministic dynamics
\[ \dot{x}=\lambda x \]
(in view of Proposition~\ref{L-expansion} this is the scalar restriction of the dynamics
$ \ket{\dot{\beta}}=B D B^{-1} \ket{\beta} $ from Eq.~\eqref{betadyn}, with
$x$ being the projection of the current state on the eigenvector $v$),
by stochastic dynamics on the probability distribution (denoted $F$) over $x$:
\be \frac{\partial F(t,x)}{\partial t} = \frac{\kappa^2}{2} \frac{\partial^2}{\partial^2 x} F(t,x) -
 \lambda \frac{\partial}{\partial x} (xF(t,x))
\label{OU} \ee 
where $\kappa$ measures the intensity of the noise.

Eq.~\eqref{OU} is the Fokker-Planck PDE describing the
 Ornstein-Uhlenbeck process~\cite{OU30,Risken89}. $F(x,t)$ converges to the Gaussian density
\[ \sqrt{\frac{-\lambda}{\pi \kappa^2}} \;  e^{\lambda x^2/\kappa^2}
\]
 Thus the stationary distribution along the $v$ eigenvector is Gaussian about the origin with standard deviation $\sqrt{-\kappa^2/2\lambda}$. Considering the slowest mode of the system, we have shown Theorem~\ref{steady}. (The clause bounding prices in the theorem is necessary because the fixed noise rate is being applied in the $\beta$ basis and the operator norm of that change of basis is bounded in terms of the range of prices, as we see from Eq.~\eqref{Bagain}.)

Normalizing the noise rate  to $\kappa^2/2=1$, we conclude that
 in a market with damping rate $\lambda$, at any given time the prices are likely to be deviating by about $\sqrt{-1/\lambda}$ from their equilibrium values. Note that, fixing any $\delta$, even for a network in which all $C_{ij}$ are $0$ or $1$ and every node has only a constant number of neighbors, $-\lambda$ may range between a constant (independent of $n$) and inverse-quadratic in $n$. Consequently the quality of connectivity of the market, even in markets of such limited form, may have an effect as large as linear in the number of agents, upon the size of typical price deviations from equilibrium.

This confirms and quantifies the expectation that in a rapidly
 damping ($\lambda \ll 0$) market the equilibrium solution is highly predictive; while it is less so in a slowly damping ($\lambda $ only slightly less than $0$) market.

\bibliographystyle{plain}
\bibliography{refs}
\end{document}